\documentclass[sigplan,screen]{acmart}
\AtBeginDocument{%
  }

\setcopyright{acmlicensed}
\copyrightyear{2026}
\acmYear{2026}
\acmDOI{XXXXXXX.XXXXXXX}
\acmConference{}

\usepackage[normalem]{ulem}

\usepackage[inline]{enumitem} 
\RequirePackage{tikz} 
\tikzset{%
  point/.style={fill=#1, inner sep=1.4pt, circle},
  point/.default=black}
\usetikzlibrary{decorations.pathreplacing, calligraphy, arrows.meta}
\RequirePackage{subcaption} 
\RequirePackage{tabularray} 
\UseTblrLibrary{booktabs}                                                       
\newcommand{\RR}{\mathbb{R}}
\newcommand{\Mat}[3]{{#1}^{#2\times #3}}
\newcommand{\MatGr}[2][\relax]{%
\operatorname{GL}_{#2}\ifx#1\relax\relax\else(#1)\fi}

\newcommand{\CV}[2]{{#1}^{#2}}
\newcommand{\ID}[1][\relax]{\mathbf{1}\ifx#1\relax\relax\else_{#1}\fi}
\newcommand{\ZERO}[1][\relax]{\mathbf{0}\ifx#1\relax\relax\else_{#1}\fi}
\newcommand{\Ortho}[2][\relax]{%
  \mathcal{O}_{#2}\ifx#1\relax\relax\else(#1)\fi}
\newcommand{\Perms}[2][\relax]{
  \Pi_{#2}\ifx#1\relax\relax\else(#1)\fi}
\newcommand{\eg}{e.\,g.}
\newcommand{\ie}{i.\,e.}
\newcommand{\wrt}{w.\,r.\,t.}
\DeclareMathOperator{\diag}{diag}
\let\vec\relax
\DeclareMathOperator{\vec}{vec}

\definecolor{ultramarine}{HTML}{5D8CAE} 
\definecolor{whitemouse}{HTML}{B9A193} 
\definecolor{crimson}{HTML}{C91F37} 
\definecolor{johannes}{HTML}{A4345D} 

\usepackage{stfloats}

\AtEndPreamble{%
  \theoremstyle{acmdefinition}
  \newtheorem{remark}[theorem]{Remark}}

\begin{document}

\title{Automatic Generation of Polynomial Symmetry Breaking Constraints}

\author{M\u{a}d\u{a}lina Era\c{s}cu}
\email{madalina.erascu@e-uvt.ro}
\orcid{0000-0002-0435-5883}
\affiliation{%
  \institution{Faculty of Informatics, West University of Timișoara}
  \city{Timișoara}
  \country{Romania}
}

\author{Johannes Middeke}
\email{johannes.middeke@tuj.temple.edu}
\orcid{0009-0007-6716-2433}
\affiliation{%
  \institution{Temple University, Japan Campus}
  \city{Tokyo}
  \country{Japan}}
\renewcommand{\shortauthors}{Era\c{s}cu and Middeke}

\begin{abstract}
Symmetry in integer programming causes redundant search and is often
handled with symmetry breaking constraints that remove as many
equivalent solutions as possible. We propose an algebraic method which
allows to generate a random family of polynomial inequalities which
can be used as symmetry breakers. The method requires as input an
arbitrary base polynomial and a group of permutations which is
specific to the integer program. The computations can be easily
carried out in any major symbolic computation software. In order to
test our approach, we describe a case study on near half-capacity 0-1
bin packing instances which exhibit substantial symmetries. We
statically generate random quadratic breakers and add them to a
baseline integer programming problem which we then solve with
Gurobi. It turns out that simple symmetry breakers, especially
combining few variables and permutations, most consistently reduce
work time.
\end{abstract}



\maketitle

\section{Introduction}\label{sec:introduction}
While symbolic computation exploits the symmetry of the problem for solving it \cite{hubert2024preserving},
optimization and automated reasoning (constraint programming~\cite{GentPetriePuget2006SymmetryCP,Walsh2012SymmetryBreakingConstraints}, mixed-integer programming~\cite{Margot2002PruningIsomorphism,kaibel2008packing,salvagnin2018symmetry}, satisfiability~\cite{AloulMarkovSakallahDAC2003Shatter,CrawfordEtAlKR1996}, and satisfiability modulo theory~\cite{DingliwalEtAlArXiv2019SMTSymBreak}) typically view the same symmetries as redundant branches of the search tree that must be pruned by explicit symmetry breaking constraints. In order to overcome this issue, a common and effective method to deal with symmetry is
to add constraints which eliminate as many symmetric solutions as possible. These are the so-called symmetry breaking constraints which can be exploited either statically~\cite{puget1993satisfiability,DingliwalEtAlArXiv2019SMTSymBreak} or dynamically~\cite{backofen1999excluding,10.5555/3006433.3006560,fahle2001symmetry}. 

Our work concerns \emph{static symmetry breaking}, i.e., adding constraints \emph{a priori} to restrict feasibility to orbit representatives. This is a mature topic: in constraint programming~(CP), lex-leader and related ordering constraints are standard~\cite{Walsh2012SymmetryBreakingConstraints,GentPetriePuget2006SymmetryCP}. In mixed-integer programming (MIP), the most structured static approaches are polyhedral: for assignment-matrix formulations with interchangeable rows/columns, orbitopes encode the convex hull of lex-sorted $0$-$1$ matrices and yield strong linear descriptions and compact extended formulations~\cite{kaibel2008packing,faenza2009extended}; more generally, static inequalities can also be derived from computational group theory (stabilizer chains, Schreier-Sims tables) to break symmetry in arbitrary formulations~\cite{LibertiOstrowski2014StabilizerBased,salvagnin2018symmetry}.
In satisfiability modulo theory (SMT), symmetry breaking is typically imposed on the Boolean abstraction~\cite{DingliwalEtAlArXiv2019SMTSymBreak}.

While the existing research on symmetry breaking has concentrated on
linear inequalities, in this work we are exploring the computational effectiveness of non-linear symmetry breakers. Starting
with a base polynomial in the problem variables, we construct
inequalities using the symmetry group of the problem. In more detail, given $h(x)$, a polynomial, and $P$, a permutation of variables in the
symmetry group, we show that the inequality
$h(P x) - h(x) \leq 0$ can act as a static symmetry breaker (see 
~\autoref{thm:MainResult}). Note that
for the special case that $h(x)$ is a linear polynomial, our approach
will generate symmetry breakers whose are  similar to those found by existing approaches such as
\cite{50Years}. However, our
method is not restricted to the linear case.

As a concrete case study, apply our polynomial symmetry breaking method on the classical 0-1 bin packing model and focus on near half-capacity instances, where item sizes are concentrated around $B/2$, where $B$ is the bin size. This is motivated by temporal bin packing with two jobs per bin (TBP2) in VM-to-server assignment and is computationally challenging due to many equivalent pairings and bin-label permutations, which leads to a high degree of symmetries~\cite{muir2024temporal}.

To evaluate the proposed symmetry breaking method, we generate
synthetic near half-capacity bin packing benchmark
families. We then create a family of random
  polynomial symmetry breaking constraints for each benchmark by
  generating a base polynomial from a selection of templates and
  applying a random set of variable permutations from the symmetry
  group of the benchmark.


Experiments were conducted with Gurobi 12.0.3 \cite{gurobi} on an
Apple M2 Pro system, measuring performance via Gurobi work units
(deterministic computational effort) under a limit of $1800$ work
units.  The results revealed that quadratic symmetry breakers
consistently outperform linear breakers and, importantly, also surpass
over Gurobi’s built-in symmetry breaking. The most reliable gains
come from small-scale quadratic breaker sets (few variables, few
permutations), with the $xy$ template performing best overall. In
contrast, purely linear breakers often increase solver effort, and
larger breaker sets show more variable behavior due to the overhead of
additional non-linear constraints.

\smallskip

The contributions of this paper are as follows:
\begin{itemize}
    \item We propose a symmetry breaking method that generates constraints of the form $h(Px)-h(x)\!\le\!0$ for group permutations $P$ and a chosen base polynomial~$h$. Unlike standard CP/MIP approaches based on linear ordering or polyhedral constraints, our construction naturally yields non-linear symmetry breakers.
    \item Our method breaks symmetry by comparing polynomial evaluations under group permutations; the resulting quadratic templates introduce non-linear interactions rather than re-encoding existing approaches into a symbolic computation framework. Hence, our approach is not incremental, but novel. 
    \item Our method is the first method generating non-linear symmetry breakers automatically---provided that (a~subgroup of) the symmetry group is known. Moreover, the experimental results show empirical evidence that quadratic breakers are efficient. 
    \item Our method is particularly valuable for problems where symmetry breakers are not known a priori\footnote{It is well-known that bin packing admits strong, well-studied and computationally effective linear symmetry breakers (e.g., lex-based representatives).}, motivating future evaluations beyond bin packing.
\end{itemize}

\smallskip

The paper is structured as follows. \autoref{sec:MainResult} presents the main symmetry breaking result for integer programs. \autoref{sec:CaseStudy} instantiates it in the bin packing model (benchmarks, symmetry subgroup, and symmetry breaker generation). \autoref{sec:ExperimentalResults} reports the experimental evaluation. \autoref{sec:ConclusionsFutureWork} concludes the paper and outlines future work. 

\section{Main Result}\label{sec:MainResult}

\begin{figure*}[t!]
  \centering
  \begin{subfigure}{0.3\textwidth}
    \centering
    \begin{tikzpicture}
      \definecolor{fillColor}{RGB}{255,255,255}
      \definecolor{brewerGreen}{RGB}{27,158,119}
      \definecolor{brewerOrange}{RGB}{217,95,2}
      \definecolor{brewerPurple}{RGB}{117,112,179}
      \definecolor{regionColour}{named}{brewerGreen}
      \definecolor{pointColour}{named}{brewerPurple}
      \fill[regionColour!15] (-1,-1) -- (1,1) -- (3,-1) -- cycle;
      \fill[regionColour!15] (-1,3) -- (1,1) -- (3,3) -- cycle;
      \draw[regionColour] (3,-1) -- (-1,3);
      \draw[regionColour] (-1,-1) -- (3,3);
      \draw[-Latex] (-1,0) -- (3,0) node[right] {$x$};
      \draw[-Latex] (0,-1) -- (0,3) node[left] {$y$};
      \node[point=pointColour,
      label=below:{\color{pointColour}$(x_1,y_1)$}] at (1,0) {};
      \node[point=pointColour,
      label=left:{\color{pointColour}$(x_2,y_2)$}] at (0,1) {};
    \end{tikzpicture}
  \end{subfigure}
  \hfill
  \begin{subfigure}{0.3\textwidth}
    \centering
    \begin{tikzpicture}
      \definecolor{fillColor}{RGB}{255,255,255}
      \definecolor{brewerGreen}{RGB}{27,158,119}
      \definecolor{brewerOrange}{RGB}{217,95,2}
      \definecolor{brewerPurple}{RGB}{117,112,179}
      \definecolor{regionColour}{named}{brewerGreen}
      \definecolor{pointColour}{named}{brewerPurple}
      \begin{scope}
        \clip (-1,-1) rectangle (3,3);
        \fill[regionColour!15] (-1,-1) -- (3,3) |- cycle;
        \fill[white, rotate=45] (0,0) ellipse (1.412 and 3);
        \clip (-1,-1) -- (3,3) -| cycle;
        \fill[regionColour!15, rotate=45] (0,0) ellipse (1.412 and 3);
      \end{scope}
      \begin{scope}
        \clip (-1,-1) rectangle (3,3);
        \draw[regionColour, rotate=45] (0,0) ellipse (1.412 and 3);
        \draw[regionColour] (-1,-1) -- (3,3);
      \end{scope}
      \draw[-Latex] (-1,0) -- (3,0) node[right] {$x$};
      \draw[-Latex] (0,-1) -- (0,3) node[left] {$y$};
      \node[point=pointColour,
      label=below:{\color{pointColour}$(x_1,y_1)$}] at (1,0) {};
      \node[point=pointColour,
      label=left:{\color{pointColour}$(x_2,y_2)$}] at (0,1) {};
    \end{tikzpicture}    
  \end{subfigure}
  \hfill
  \begin{subfigure}{0.3\textwidth}
    \centering
    \begin{tikzpicture}
      \definecolor{fillColor}{RGB}{255,255,255}
      \definecolor{brewerGreen}{RGB}{27,158,119}
      \definecolor{brewerOrange}{RGB}{217,95,2}
      \definecolor{brewerPurple}{RGB}{117,112,179}
      \definecolor{regionColour}{named}{brewerGreen}
      \definecolor{pointColour}{named}{brewerPurple}
      \begin{scope}
        \clip (-1,-1) rectangle (3,3);
        \fill[regionColour!15] (-1,-1) -- (3,3) |- cycle;
      \end{scope}
      \begin{scope}
        \clip (-1,-1) rectangle (3,3);
        \draw[regionColour] (-1,-1) -- (3,3);
      \end{scope}
      \draw[-Latex] (-1,0) -- (3,0) node[right] {$x$};
      \draw[-Latex] (0,-1) -- (0,3) node[left] {$y$};
      \node[point=pointColour, 
      label=below:{\color{pointColour}$(x_1,y_1)$}] at (1,0) {};
      \node[point=pointColour, 
      label=left:{\color{pointColour}$(x_2,y_2)$}] at (0,1) {};
    \end{tikzpicture}    
  \end{subfigure}
  \caption{The regions $2 y + x^2 - 2 x - y^2 \leq 0$ (left) and
    $y^3 - 3 y + 3 x - x^3 \leq 0$ (centre) and $y - x \leq 0$ (right)
    are shown in a darker shade. In each case, one optimal solution of
    the example problem is inside the region while the other one is
    not.}
  \label{fig:geomshapesbs}
\end{figure*}
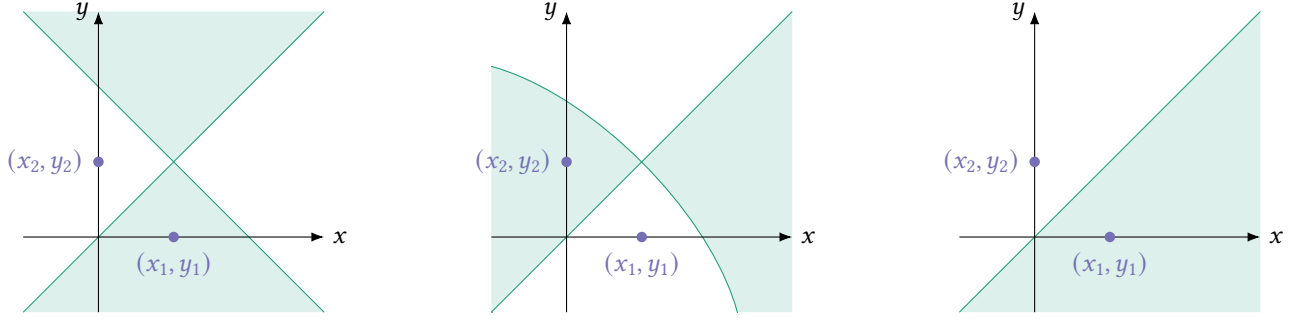

We consider integer programming problems of the form
\begin{equation}
  \label{eq:IP}\tag{IP}
    \begin{aligned}
    \text{minimise}\quad & f(x) \\
    \text{subject to}\quad & A x \geq b, \\
    & x \in U^n
  \end{aligned}
\end{equation}
where
\begin{enumerate*}[afterlabel={},label={},itemjoin={\space}]
\item $f \colon \CV{\RR}{n} \to \RR$ is a polynomial function,
\item $A \in \Mat{\RR}{m}{n}$ is a matrix,
\item $b \in \CV{\RR}m$ is a vector, 
\item $x$ is a vector of decision variables, and
\item $U \subset \RR$ is a finite set.
\end{enumerate*}
See, \eg,~\cite{50Years}.

Integer programs are typically solved with tree-search methods such as branch-and-bound, which explore a search tree by branching on variable assignments and pruning subproblems using bounds. If the model has symmetries, then many branches correspond to permuted, equivalent solutions, so the solver may repeatedly explore essentially identical subtrees. This redundancy can greatly increase the search effort, motivating symmetry breaking constraints that restrict the solver to one representative per symmetry orbit.

\smallskip

Since there are many subtly different definitions of symmetries in the
literature (see, \eg, \cite{GentPetriePuget2006SymmetryCP} for an
overview), it seems prudent to clarify the exact meaning which is used
in this paper. Just as the problem description~\eqref{eq:IP}, it is
based on the definition given in \cite{50Years}.

\begin{definition}[Symmetry]
  An $n\times n$ permutation matrix $P$ is called a \emph{symmetry}
  of~\eqref{eq:IP} if we have $f(P x) = f(x)$ and
  $A (P x) \geq b \iff A x \geq b$ for all $x\in U^n$. The set of all
  symmetries for a given problem forms a group under multiplication
  which we refer to as the \emph{symmetry group} of~\eqref{eq:IP}.
\end{definition}

We will denote the symmetry group of~\eqref{eq:IP} by $G$. Moreover,
we will write $\Perms{n}$ for the set of all $n\times n$ permutation
matrices.

If $x^* \in U^n$ is an optimal solution of~\eqref{eq:IP} and if
$P \in G$ is a symmetry of~\eqref{eq:IP}, then also $P x^*$ is an
optimal solution. This introduces a redundancy in the search. Hence,
one will commonly try to avoid generating solutions which are
connected to other solutions by a symmetry.

\begin{definition}[Symmetry Breaker]
  A \emph{symmetry breaker} is a predicate
  $\psi\colon U^n \to \{\mathit{true}, \mathit{false}\}$ in the
  decision variables such that there exist an optimal solution
  $x^* \in U^n$ and a symmetry $P \in G$ with
  $\psi(x^*) = \mathit{true}$ and $\psi(P x^*) = \mathit{false}$.
\end{definition}

In static symmetry breaking, symmetry breakers are added a priori to the problem~\eqref{eq:IP}. More formally, given a symmetry breaker
$\psi(x)$, we obtain a new problem
\begin{equation}
  \label{eq:IP'}\tag{IP'}
  \begin{aligned}
    \text{minimise}\quad & f(x) \\
    \text{subject to}\quad & A x \geq 0, x \in U^n, \psi(x).
  \end{aligned}
\end{equation}
Since $\psi(x)$ is fulfilled for at least one optimal solution
of~\eqref{eq:IP}, solving~\eqref{eq:IP'} is sufficient for finding a
solution for~\eqref{eq:IP}.

It is possible to add more than one symmetry breaker to the
problem. However, in that case care must be taken that the symmetry
breakers are compatible with each other, \ie, we must ensure that
there is at least one optimal solution which is fulfilled by all the
added symmetry breakers simultaneously. Otherwise, the symmetry
breakers would render the problem infeasible.
\medskip

While symmetry breakers could in principle take any form, in this work
we will concentrate on polynomial inequalities. Consider the
polynomial ring $R = \RR[x]$ where $x~=~(x_1, \ldots, x_n)$ is a
vector of~$n$ variables. For a polynomial $h(x) \in R$ and a
permutation $P \in G$, we will write $h(P x)$ for the polynomial which
arises from $h(x)$ by substitution of all the variables according to
$P$. More precisely, if $P$ maps the $i$-th unit vectors $e_i$ to the
$j$-th unit vector $e_j$, then we subsitute $x_j$ for $x_i$ in
$h(P x)$.

\begin{theorem}\label{thm:MainResult}
  Let $G$ be the symmetry group of~\eqref{eq:IP} and let $R = \RR[x]$
  where $x = (x_1, \ldots, x_n)$. Then, for any $h(x) \in R$ and any
  $P \in G$, there is at least one optimal solution $y^* \in U^n$
  which fulfils the inequality
  \begin{equation*}
    h(P x) - h(x) \leq 0.
  \end{equation*}
  That means, $h(P x) - h(x) \leq 0$ can be added to the
  problem~\eqref{eq:IP} as a symmetry breaker without losing the
  ability to find an optimal solution.
\end{theorem}

\begin{proof}
  Assume that $x^* \in U^n$ is an optimal solution
  of~\eqref{eq:IP}. Then any element of the orbit
  $G x^* = \{ P x^* \mid P \in G\}$ is also an optimal solution.
  Moreover, $G x^*$ is finite and so it the set
  $H = \{h(y^*) \mid y^* \in G x^* \} \subseteq \RR$. Thus, $H$ has a
  maximum. Let $y^* \in G x^*$ be such that $h(y^*) = \max H$. Then,
  we have
  \begin{equation*}
    \forall P \in G\colon
    h(P y^*) \leq h(y^*).
    \qedhere
  \end{equation*}  
\end{proof}

\begin{remark}
  For any fixed polynomial $h \in R$, \autoref{thm:MainResult} allows
  us to use as many permutations $P_1, \ldots, P_\ell \in G$ as we
  like in order to generate an entire family of symmetry breakers
  \begin{equation*}
    h(P_1 x) - h(x) \leq 0, \quad
    \ldots, \quad
    h(P_\ell x) - h(x) \leq 0.
  \end{equation*} 
  Moreover, adding all breakers from the same family to the problem
  simultaneously will not impede our ability to find an optimal
  solution because the theorem ensures that there is at least one
  solution which fulfils all of the breakers.
\end{remark}

It is important to note that the choice of the base polynomial
$h \in R$ has huge effect on the effectiveness of the breakers. In the
extreme case, if $h$ is invariant under the group $G$, then the
breakers will just be the trivial inequality $0 \leq 0$. It is also
important to note that $G$ will be very large in general. Hence, while
in theory we compute all possible breakers for any given base
polynomial, in practise this is unfeasible.

\medskip

\begin{example}
  To illustrate \autoref{thm:MainResult}, consider the two variable
  problem
  \begin{equation*}
    \begin{aligned}
      \text{minimise}\quad& x + y \\
      \text{subject to}\quad& x + y \geq 1, \\
      & x, y \in \{0,1\}.
    \end{aligned}
  \end{equation*}
  For this problem, the permutation group is
  $G\!\!=\!\{\ID, M_{(1\,2)}\!\}\!\subseteq\!\Perms{2}$ where $M_{(1\,2)}$ is
  the permutation matrix for the transposition $(1\,2) \in S_2$. We
  randomly choose the base polynomial $h~=~2 x + y^2 \in
  \RR[x,y]$. This yields the symmetry breaker
  \begin{multline*}
    h(M_{(1\,2)}(x,y)^t) - h(x,y)
    = h(y,x) - h(x,y)
    \\
    = 2 y + x^2 - 2 x - y^2
    \leq 0.
  \end{multline*}
  The two optimal solutions for the example problem are
  $(x_1,y_1) = (1,0)$ and $(x_2,y_2) = (0,1)$. We can check that
  $(x_1,y_1)$ fulfils the inequality $2 y + x^2 - 2 x - y^2 \leq 0$
  while $(x_2,y_2)$ does not (see \autoref{fig:geomshapesbs}).

  As another base polynomial, consider $h = x^3 - 3 x$. For this case,
  the inequality becomes
  \begin{equation*}
    y^3 - 3 y + 3 x - x^3 \leq 0.
  \end{equation*}
  Again, we can check that one of the solutions, more precisely
  $(x_2,y_2)=(0,1)$, fulfills this new inequality while the other one does not (see also \autoref{fig:geomshapesbs}).

  Finally, consider the linear base polynomial $h = 2 x + y$. Here, the
  inequality is
  \begin{equation*}
    2 y + x - 2 x - y = y - x \leq 0.
  \end{equation*}
  (Once more, \autoref{fig:geomshapesbs} illustrates the region where the
  inequality is fulfilled.)
\end{example}

\begin{remark}
  If we start with a non-zero linear base polynomial
  $h(x) = v_1 x_1 + \ldots + v_n x_n$, then our method will produce
  linear symmetry breakers. In fact, it is not difficult to show that
  we will obtain exactly the same linear symmetry breakers as ones
  obtained for the vector $\bar{x} = (v_1, \ldots, v_n)$ in
  ~\cite[Thm.~17.3]{50Years}. However, our method is more general since
  it is not restricted to linear polynomials only. In particular, as~\autoref{fig:geomshapesbs} shows, we obtain a much larger variety of
  different geometric shapes if we allow non-linear symmetry breakers.
\end{remark}

The proof which is given for the correctness of the linear symmetry
breakers in \cite{50Years} relies on the concept of fundamental
regions as introduced in \cite{GroveBenson1985}.  Even though the
proof of our \autoref{thm:MainResult} is much simpler and does not
need this extra machinery, it is possible to show that our symmetry
breakers are also capable of defining fundamental regions. We repeat
the definition here for the convenience of the reader.

\begin{definition}[{\cite[Chpt.~3]{GroveBenson1985}}]
  A set $F \subseteq \CV{\RR}{n}$ is a \emph{fundamental region} \wrt\
  a finite subgroup $G \leq \Ortho{n}$ if
  \begin{enumerate}
  \item $F$ is open,
  \item $P F \cap F = \emptyset$ for all
    $P \in G \setminus \{\ID[n]\}$,
  \item $\bigcup_{P \in G} \overline{P F} = \CV{\RR}{n}$.
  \end{enumerate}
\end{definition}

For the theorem we restrict ourselves to the case of groups of
permutation matrices. However, this restriction is not essential for this result.

\begin{theorem}\label{thm:FundamentalRegion}
  Let $G \leq \Perms{n}$ be a finite group of permutations. Let
  $h \in \RR[x]$. For each $P \in G$, define
  \begin{equation*}
    L_P = \{ x \in \CV{\RR}{n} \colon h(P x) - h(x) < 0 \}.
  \end{equation*}
  If $L_P \neq \emptyset$ for each $P \in G$ with $P \neq \ID[n]$,
  then $F = \bigcap_{P \in G \setminus \{\ID\}} L_P$ is a fundamental
  region \wrt~$G$.
\end{theorem}

The proof is inspired by the one of
\cite[Thm.~3.1.2]{GroveBenson1985}. However, while
\cite{GroveBenson1985} makes a heavy use of linear algebra concepts, in particular norms and orthogonal maps, our proof is not restricted to the linear case.

\begin{proof}[Proof of \autoref{thm:FundamentalRegion}]
  We first note that $L_P$ is open since $x \mapsto h(P x) - h(x)$ is
  a continuous function. This implies that $F$ is open since it is the
  intersection of finitely many open sets.

  Note that
  \begin{math}
    F = \{x\in \CV{\RR}{n}\colon h(Q x) < h(x) \;\forall Q \in G, Q \neq \ID
    \}.
  \end{math}
  Thus, we have
  \begin{multline*}
    P F
    = P (\{x \colon h(Q x) < h(x) \;\forall Q \neq \ID \})
    \\
    = \{ P x \colon h(Q x) < h(x) \;\forall Q \neq \ID \}
    \\
    = \{ y = P x \colon h(Q P^{-1} y) < h(P^{-1} y) \;\forall Q \neq \ID \}
    \\
    = \{ y \colon h(U y) < h(P^{-1} y) \;\forall U \neq P^{-1} \} 
  \end{multline*}
  since $U = Q P^{-1}$ will range over all permutations except for
  $P^{-1}$ as $Q$ ranges over the permutations in
  $G \setminus \{\ID\}$.  Consequently, if we had $x \in F \cap P F$,
  then we obtained the contradiction
  \begin{math}
    h(x) < h(P^{-1} x) < h(x)
  \end{math}
  where the first inequality is from $x \in P F$ with $U = \ID$ and
  the second is from $x \in F$. Thus, $F \cap P F = \emptyset$.

  Finally, 
  \begin{math}
    \overline{P F}
    = \{y \colon h(U y) \leq h(P^{-1} y) \;\forall U \neq P^{-1}\}
  \end{math}
  for every $P \in G$. Let $x \in \CV{\RR}{n}$ be arbitrary. Since $G$ is
  finite, so is $M = \{ P x \colon P \in G \}$. Hence, there is a
  $P \in G$ such that $h(P^{-1} x) = \max M$. Then
  \begin{math}
    x \in \overline{P F} \subseteq \bigcup_{Q \in G} \overline{Q F}.
  \end{math}
\end{proof}

\begin{remark}
  By \cite[Thm.~3.1.2]{GroveBenson1985}, there always exists a linear
  polynomial $h = a_1 x_1 + \ldots + a_n x_n \in \RR[x]$ with
  $a_1, \ldots, a_n \in \RR$ which fulfils the condition
  $L_P \neq \emptyset$ for all $P \in G \setminus \{\ID[n]\}$
  in~\autoref{thm:FundamentalRegion}.
\end{remark}

\section{Case Study}\label{sec:CaseStudy}
\subsection{Problem Description (Bin Packing Model)} 
The bin packing problem \cite[Chapter~8]{martello1990knapsack} is given by
\begin{equation}
  \label{eq:BP}\tag{BP}
  \begin{aligned}
    \text{minimise}\quad
    &
    y_1 + \ldots + y_n \\
    \text{subject to}\quad
    &
    \forall k=1,\ldots,n\colon
    s_1 x_{1k} + \ldots + s_m x_{mk} \leq B y_k \\
    &
    \forall i=1,\ldots,m\colon
    x_{i1} + \ldots + x_{in} = 1 \\
    &
    \forall i = 1,\ldots,m\; 
    \forall k = 1,\ldots,n\colon
    x_{i,k} \in \{0,1\} \\
    &
    \forall k = 1,\ldots,n\colon
    y_k \in \{0,1\} \\
  \end{aligned}
\end{equation}
where $m, n$ are positive integers and 
$B$, $s_1, \ldots, s_m$, are positive real numbers.
We will refer to
$s_1, \ldots, s_m$ as the (object) \emph{sizes} of~\eqref{eq:BP} and
to $B$ as the \emph{bin size}.
\subsection{Benchmark Generation (Near Half-Capacity Families)} Inspired by \cite{muir2024temporal}, we evaluate our method on a static bin-packing benchmark of near half-capacity instances, where item sizes are concentrated around $B/2$ (with $B$ the bin capacity). This setting is practically motivated by temporal bin packing with two jobs per bin (TBP2), which arises in VM-to-server assignment for specialized resource-intensive cloud systems. Computationally, such instances are challenging as they exhibit many equivalent item pairings, arbitrary bin-label permutations, and almost all objects are undistinguishable.

We construct a collection of synthetic instance families for the bin-packing problem~\eqref{eq:BP}. For each family, we report the bin capacity $B$, the number of items $n$, and the item-size specification (here, a near half-capacity regime).\autoref{tab:near-half-capacity} summarizes the settings selected to yield a runtime for the baseline problems of approximately $5$ minutes. The generated benchmarks have around 4 million variables and around 4000 constraints.
\begin{table}[h!]
\centering
\caption{Near half-capacity experimental settings}
\label{tab:near-half-capacity}
\small
\centering
\begin{tabular}{lcccc}
\hline
Class & 3 & 5 & 7 & 9 \\
\hline
$n$ & 2000 & 2000 & 1024 & 1000 \\
intervals & [49,51] & [48,52] & [47,53] & [46,54] \\
\hline
\end{tabular}
\end{table}


\subsection{Symmetry Group}

For the symmetry breakers, we will first need to determine (a large
enough subgroup of) the symmetry group for the bin packing
problem.

We will write $e_1 = (1,0,\ldots,0), e_2 = (0,1,0,\ldots,0), \ldots$
for the canonical basis vectors (of arbitrary length) and
$e = (1, \ldots, 1)$ for the all ones vector (again, of arbitrary
length). Moreover, we will write $\ID[q]$ for the $q\times q$ identity
matrix. We will denote the symmetric group over $\{1,\ldots,q\}$ by
$S_q$. For $\tau \in S_q$, we use $M_\tau$ for the $q\times q$
permutation matrix with $M_\tau e_{\tau(i)} = e_{i}$ for all
$i=1,\ldots,q$. Further, for $\tau \in S_q$ we define
$\hat{M}_\tau = \diag(1, M_\tau) \in \MatGr[\RR]{q+1}$.

It will be convenient to collect the variables of~\eqref{eq:BP} in an
$(m+1)\times n$ matrix
\begin{equation*}
  V = \begin{pmatrix}
    y_1 & \cdots & y_n \\
    \cmidrule(lr){1-3}
    x_{11} & \cdots & x_{1n} \\
    \vdots &        & \vdots \\
    x_{m1} & \cdots & x_{mn} \\
  \end{pmatrix}.
\end{equation*}
This allows us to represent the objective function as well as the
conditions using Kronecker products and the vectorisation
function. (See, \eg, \cite[Sect.~11.4]{HBofLA} for the notation and
its basic properties.) This allows us to write the objective function
as
\begin{math}
  (e \otimes e_1) \vec(V).
\end{math}
Similarly, with $u = (B, -s_1, \ldots, -s_m)$ we are able to write the
constraints as
\begin{math}
  (\ID[n] \otimes u) \vec(V) \geq 0
\end{math}
and
\begin{math}
  (e \otimes (0\,\ID[m])) \vec(V) = e.
\end{math}
(Here, $(0\, \ID)$ denotes a column of zeroes followed an identity
matrix.)

In the experiments in \autoref{sec:MainResult} we  consider two kinds of permutation matrices. For the first kind, consider any
permutation $\tau \in S_n$. Then
\begin{multline*}
  (e \otimes e_1) (M_\tau \otimes \ID[m+1]) \vec(V)
  =
  (e M_\tau \otimes e_1) \vec(V)
  \\
  =
  (e \otimes e_1) \vec(V).
\end{multline*}
Thus, applying the permutation matrix $M_\tau \otimes \ID[m+1]$ to the
variables vector $\vec{V}$ will not change the objective
function. Similarly,
\begin{math}
  (e\otimes (0\,\ID[m])) (M_\tau \otimes \ID[m+1])
  =
  (e\otimes (0\,\ID[m]))
\end{math}
which means that $(M_\tau \otimes \ID[m+1])$ also leaves the second
condition unchanged. Finally,
\begin{multline*}
  (\ID[n] \otimes u) (M_\tau \otimes \ID[m+1]) \vec(V) \\
  =
  M_\tau (\ID[n] \otimes u) \vec(V) \geq M_\tau 0 = 0
\end{multline*}
is equivalent to the first condition since $M_\tau$ only changes the
order of the inequalities. Thus, we can conclude that the matrices
$M_\tau \otimes \ID[m+1]$ are members of the symmetry group
of~\eqref{eq:BP}.

For the second kind of permutations, let us make the assumption that
the sizes are sorted in ascending order; \ie, that
$s_1 \leq \ldots \leq s_m$. Consider indices
$1 = i_1 < \ldots < i_{\ell + 1} = m + 1$ such that
$s_{i_j} = \ldots = s_{i_{j+1} - 1}$ for all $j = 0, \ldots, \ell$ and
$s_{i_j - 1} < s_{i_j}$ for $j=2, \ldots, \ell$. We will refer to
$i_1, \ldots, i_{\ell + 1}$ as the \emph{size boundaries} for the
given instance of the bin packing problem. Let us write
$S_{i_1,\ldots,i_{\ell+1}}$ for the subgroup of permutations
$\sigma \in S_m$ such that
\begin{math}
  \sigma(\{i_j, \ldots, i_{j+1} - 1\}) = \{i_j, \ldots, i_{j+1} - 1\}
\end{math}
for all $j=1, \ldots, \ell$. (That is, we have $s_j = s_{\sigma(j)}$
for any $\sigma \in S_{i_1,\ldots,i_{\ell+1}}$ and any
$j=1,\ldots,m$.)  For such a $\sigma \in S_{i_1,\ldots,i_{\ell+1}}$
consider the permutation matrix $(\ID[n] \otimes \hat{M}_\sigma)$. We
have
\begin{math}
  (e \otimes e_1) (\ID[n] \otimes \hat{M}_\sigma)
  =
  (e \otimes e_1)
\end{math}
which shows that $(\ID[n] \otimes \hat{M}_\sigma)$ leaves the
objective function invariant. For the first condition, we get
\begin{math}
  (\ID[n] \otimes u) (\ID[n] \otimes \hat{M}_\sigma)
  =
  (\ID[n] \otimes u)
\end{math}
since $u \hat{M}_\sigma = u$ by the assumption that
$s_j = s_{\sigma(j)}$ for all $j=1,\ldots,m$. Finally,
\begin{math}
  (e \otimes (0\,\ID[m])) (\ID[n] \otimes \hat{M}_\sigma)
  =
  M_\sigma (e \otimes (0\,\ID[m]))
\end{math}
which together with $M_\sigma e = e$ shows that also the second
condition is invariant under $(\ID[n] \otimes \hat{M}_\sigma)$. Thus,
also matrices of this kind are members of the symmetry group
of~\eqref{eq:BP}.
\subsection{Symmetry Breakers Generation}

For the near half-capacity instance families
in~\autoref{tab:near-half-capacity}, we generate linear and quadratic
symmetry breakers by instantiating a set of templates (see
\autoref{tab:symbreakers}) and varying their sizes via four subclasses
that control the number of variables and sampled permutations. For the
templates we use a shorthand notation where each $x$ stands for a
random linear polynomial containing $x_{i,k}$, $i,k=1,\ldots,n$, while
each $y$ stands for a random linear polynomial in $y_1, \ldots,
y_n$. The coefficients of these random polynomials are restricted $0$
and $1$.  Moreover, any $x^2$ or $y^2$ in the templates refers to
products of two different linear polynomials.


\begin{table}[h!]
\centering
\small
\setlength{\tabcolsep}{6pt}
\renewcommand{\arraystretch}{1.08}
\caption{Symmetry breaker templates and sizes used in our experiments.} 
\label{tab:symbreakers}

\begin{tabular}{lr}
  \toprule
  \textbf{Category} & \textbf{Templates} \\
  \midrule
  Linear         & $x$, $y$, $x + y$ \\
  Quadratic      & $x^2$, $y^2$, $x\,y$, $x^2 + y^2$ \\
  Mixed          & $x + y^2$, $x^2 + y$ \\
  \bottomrule
\end{tabular}

\begin{tabular}{@{}l r r@{}}
\toprule
\textbf{Size} & \textbf{\#vars} & \textbf{\#perms} \\
\midrule
Few variables, pew permutations  & $\approx 10$   & $50$  \\
Few variables, many permutations & $\approx 10$   & $500$ \\
Many variables, few permutations & $\approx 1000$ & $50$  \\
Numerous variables, few permutations & $\approx 4000$ & $50$  \\
\bottomrule
\end{tabular}
\end{table}






Note that the numbers of variables differ slightly between the
templates. For instance, for the $x^2$, $y^2$, and $x y$ templates the
``few variables'' case use $9$~variables (\ie, we form products of
linear polynomials with $3$-variables). For the $x^2 + y^2$ template
we use $18$ variables; $x + y^2$ and $x^2 + y$ use $16$ variables; and
$x$, $y$, and $x + y$ use $10$ variables. Similarly, the exact number
of variables in the ``many variables'' and ``numerous variables''
cases does vary as well.

For each shape and each size, we generate a random base polynomial $h$
and then we apply $50$ (or $500$) permutations $P$ from the symmetry
group in order to obtain symmetry breakers $h(P x) - h(x) \leq 0$. For
the $x^2 + y$ and $x + y^2$ shapes, we keep only those breakers which
include quadratic terms. (Without this filter, the quadratic terms
might cancel and we could end up with collection of linear
polynomials. However, we want to test linear symmetry breakers
separately.)

There are $n!$ permutations of the type $M_\tau \otimes \ID[m+1]$ with
$\tau \in S_n$. For the type $\ID[n] \otimes \hat{M}_\sigma$ with
$\sigma \in S_{i_1, \ldots, i_{\ell+1}}$ there are
\begin{math}
  \prod_{j=1}^{\ell} (i_{j+1} - i_j)
\end{math}
many options since
\begin{math}
  S_{i_1, \ldots, i_{\ell+1}} \cong \bigtimes_{j=1}^{\ell} S_{i_{j+1} - i_j}.
\end{math}
Since both types of permutations commute that gives us a total of
\begin{math}
  n!\, \prod_{j=1}^{\ell} (i_{j+1} - i_j)
\end{math}
many permutations in $G$. For our experiments, this is an
impractically large number. For instance, for the 3~classes example
with $m = n = 2000$ and size boundaries $0, 688, 1320, 2000$ we have
more than $10^{10\,520}$ permutations.

Hence, in order to generate permutations, we will follow a different
strategy. The symmetry group is generated by all (transposition)
matrices of the form
\begin{math}
  M_{(1\,k)} \otimes \ID[m+1]
\end{math}
with $k = 2, \ldots, n$ and all matrices of the form
\begin{math}
  \ID[n] \otimes \hat{M}_{(i_j\,k)}
\end{math}
where $j = 1, \ldots, \ell$ and $k = i_j + 1, \ldots, i_{j+1} -
1$. For each base polynomial $h$, we simply pick $50$ of these
generators at random (with replacement) and form their product. This
is then applied to the base polynomial~$h$.

\section{Experimental Results}\label{sec:ExperimentalResults}

\subsection{Experimental setup} We evaluated our framework on symmetry breakers generation on the near half-capacity bin-packing families (see \autoref{tab:near-half-capacity}). All experiments were run with Gurobi Optimizer 12.0.3 \cite{gurobi} (academic license) on an Apple M2 Pro (10 cores). 
\autoref{tab:GurobiSbs} summarizes the four Gurobi configurations evaluated in our experiments. All runs enable \texttt{nonconvex} quadratic handling. The \emph{Baseline} configuration disables both \texttt{Presolve} and \texttt{Symmetry} handling, serving as a reference point\footnote{According to the online documentation, \texttt{Presolve} simplifies the model. The built-in symmetry breaking methods can interfere with our symmetry breakers.}. The \emph{Default} configuration uses Gurobi automatic settings for both \texttt{Presolve} and \texttt{Symmetry} processing. To study the impact of symmetry breaking strength in Gurobi, we keep \texttt{Presolve} on automatic and vary \texttt{Symmetry}: \emph{Conservative} applies a mild symmetry strategy, whereas \emph{Aggressive} applies the strongest symmetry handling considered. 
\begin{table}[h!]
\small
  \centering
  \caption{Gurobi configurations used in our experiments}
  \label{tab:GurobiSbs}
  \begin{tabular}{lcccc}
    \toprule
    \textbf{\!\!\!Parameter\!\!}\! & \textbf{Baseline} & \textbf{Default} & \textbf{Conservative} & \textbf{\!Aggressive\!\!\!} \\
    \midrule
    \texttt{\!\!\!Nonconvex\!\!\!} & 2 & 2 & 2 & \!2\!\!\! \\
    \texttt{\!\!\!Presolve\!\!\!}  & 0 & -1 & -1 & \!-1\!\!\! \\
    \texttt{\!\!\!Symmetry\!\!\!}  & 0 & -1 & 1 & \!2\!\!\! \\
    \bottomrule
  \end{tabular}
\end{table}
Note than even with the two options above disabled, Gurobi still has enabled a series of options (e.g. cuts, heuristics) which are used before applying the search method in the tree.

We report Gurobi work units (with a deadline of $1800$), a built-in statistic that quantifies the computational effort spent on optimization and provides a deterministic alternative to wall-clock runtime\footnote{Gurobi defines work units as a deterministic measure of computational effort expended by the solver, intended to be less sensitive than runtime to machine load and hardware differences.}.

For the non-linear symmetry breakers, SMT solvers are not applicable in this setting due to their lack of support for non-linear constraints. We therefore focus on MIP solving in this first study. Investigating constraint programming approaches (e.g., CPLEX \cite{manual1987ibm}) that can handle non-linear constraints is left for future work.

Benchmark instances, generated symmetry breakers, and scripts to reproduce and independently verify the experiments are made publicly available in our GitHub repository~\cite{github}.

\subsection{Symmetry breaker construction used in the experiments} We instantiate our symmetry breaker generation method using the templates in \autoref{tab:symbreakers} and vary the size via four subclasses controlling the number of variables and sampled permutations. 

For generating the symmetry breakers, we initially used the SymPy library in Python \cite{meurer2017sympy}. However, for the input sizes of our experiments and the number of different test cases that we use, the Python code was found to be too slow. Instead, we implemented the method in the compiled general purpose programming language OCaml~\cite{ocaml}. Given the size boundaries for an instance of the bin packing problem, our custom program can generate hundreds of symmetry breaker families in between one to three minutes. The implementation can be found in the GitHub repository \cite{github}.

\subsection{Interpretation of the results}\label{sec:InterpretationOfTheResults}
We evaluated a total of $1440$ problem instances, obtained by combining multiple benchmark classes, breaker templates and sizes all on top of the Baseline configuration. The overall number of timeouts was very low: only $10$ runs (approximately $0.7\%$) reached the time limit. Of these, the vast majority occurred in settings with linear symmetry breakers, while only a single timeout was observed in a quadratic configuration. 
\begin{figure*}[h!]
   \centering
   \input{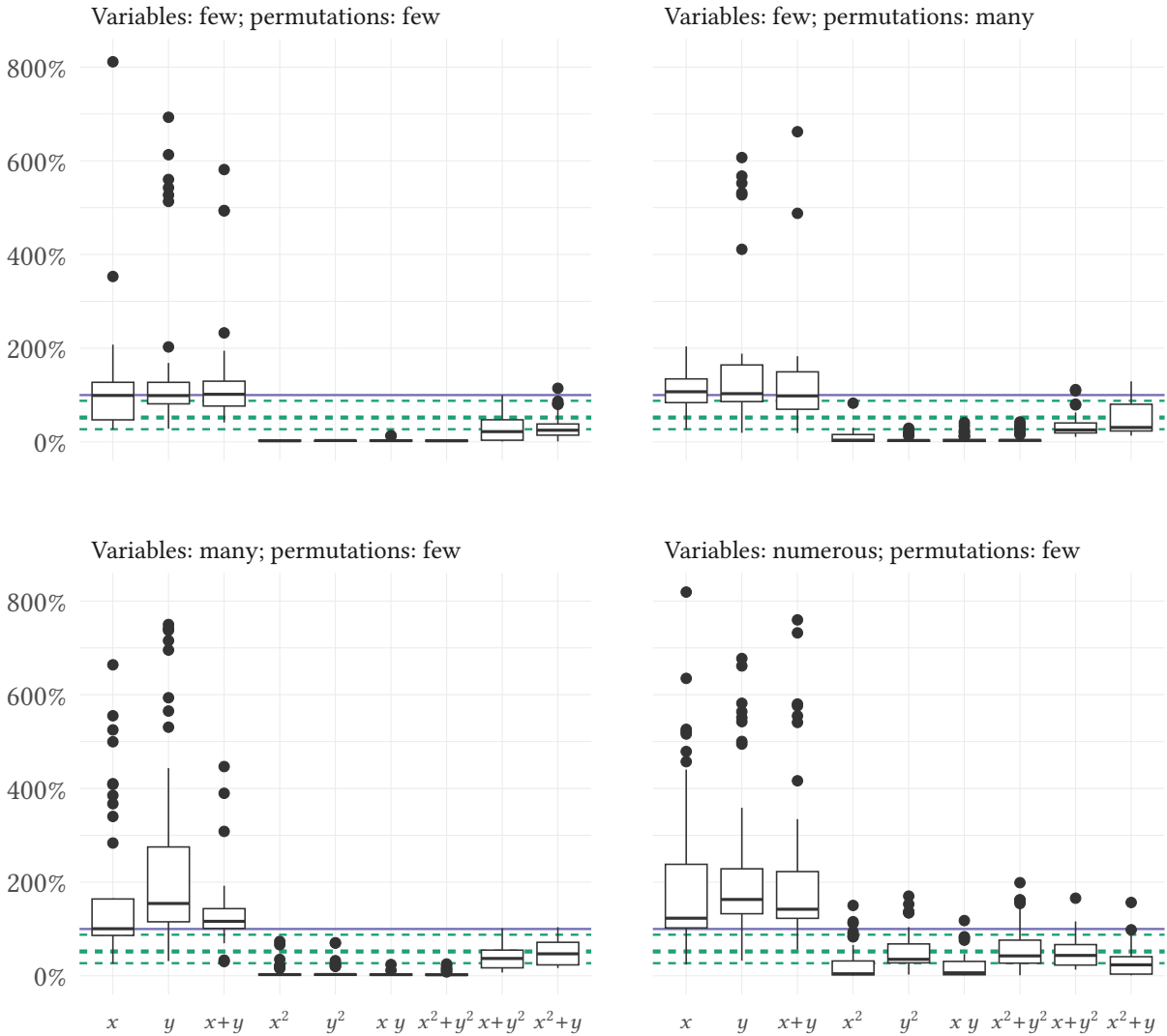}
   \caption{Combined results for all test cases (3, 5, 7, or 9 classes). The plot summarises the
     work units for each breaker shape and for each tested
     combinations of number of variables and permutations. The work
     units are scaled relative to the work units of the Baseline
     problem (\ie, the version without any breakers and with the same
     number of classes). The Baseline is shown as a solid purple line, and the relative timings for Gurobi's Default setting for each of the test cases are shown as dashed green lines.}
   \label{fig:CombinedClasses}
 \end{figure*}

\autoref{fig:CombinedClasses} aggregates the experimental results across all instance classes and summarizes the relative Gurobi work units for each symmetry-breaker template and scale. All values are normalized with respect to the Baseline setting without additional symmetry breakers, indicated by the solid purple line at 100\%. We also report the relative timings for Gurobi's Default setting for each of the classes. Those are shown as dashed green lines. The Conservative and Aggressive configurations did not complete within the allotted work units and are therefore omitted from the figure.

Several clear trends emerge. First, quadratic breakers consistently outperform linear ones. In addition, they also surpass Gurobi's built-in symmetry handling. In particular, the small-scale (\ie, few variables, few permutations) $xy$ template is superior to both, the linear breakers and Gurobi's built-in symmetry handling.

Second, small-scale symmetry breaker sets, both linear and quadratic, are the most reliable and consistently beneficial across all templates.

Third, purely linear breakers tend to be less efficient and often increase solver effort, especially when the number of variables or permutations grows. This behavior is consistent across all four size settings.

Fourth, as the scale of the symmetry breakers increases, either by involving many variables or by sampling many permutations, the performance becomes more variable. While some quadratic templates still achieve improvements, large breaker sets can also slow down the solver, indicating a trade-off between symmetry reduction and the overhead introduced by additional non-linear constraints. This effect is most pronounced in the "many variables" and "numerous variables" settings.

\section{Conclusions and Future Work}\label{sec:ConclusionsFutureWork}
We presented a static symmetry breaking approach that automatically
generates constraints. As input we only need to know (a subgroup of)
the symmetry group of the problem. We then generate a random
polynomial~$h(x)$ in the decision variables using a fixed template
and apply random permutations $P$ from the symmetry group in order to
find a family of symmetry breakering inequalities of the form
$h(P x) - h(x) \leq 0$. Since the approach is based on simple
algebraic manipulations, it can easily be carried out with any
symbolic computation software. Unlike existing methods, our technique
naturally yields non-linear symmetry breakers and introduces
non-linear interactions rather than re-encoding existing linear
schemes.  On near half-capacity bin packing, statically generated
quadratic breakers consistently improve over the baseline; the
smallest breaker size is the most reliable, and the mixed bilinear
cross-term template is the most efficient, while larger breaker sets
are more variable and can sometimes slow the solver.

As future work, we plan the followings.
\begin{enumerate*}[label=(\arabic*)]
\item Improve the permutation generation algorithm. At the moment, the
  permutations which we apply are completely random and they often do
  not change the polynomial~$h(x)$. That means that many of the
  generated inequalities are trivial and need to be discarded. Out of
  fifty attempts we would sometimes find fewer than ten non-trivial
  inequalities. This could be improved by restricting the random
  generation to permutations which do not leave the variables that
  occur in $h(x)$ fixed.
\item Evaluate the approach also to third degree symmetry breakers. This is a programming exercise as we need to generate the symmetry breakers in the format of expression trees \cite{Preiss1998ExpressionTrees} which is the way Gurobi expresses cubic constraints. 
\item Evaluate our method on arbitrary functions (\eg, algebraic or
  transcendental functions). 
\item Evaluate constraint programming solvers such as IBM ILOG CPLEX Optimization Studio \cite{manual1987ibm}.
\item We also plan to extend our symmetry breaking generation approach to scenarios in which bin items are not independent. Such cases commonly arise in the process of deployment in the Cloud of component-based applications and microservice systems, where deployment requires respecting inter-component constraints.
\end{enumerate*}

\bibliographystyle{ACM-Reference-Format}
\bibliography{symmetry}

@article{Walsh2012SymmetryBreakingConstraints,
  author  = {Toby Walsh},
  title   = {Symmetry Breaking Constraints: Recent Results},
  journal = {arXiv preprint arXiv:1204.3348},
  year    = {2012}
}

@article{kaibel2008packing,
  title={Packing and Partitioning Orbitopes},
  author={Kaibel, Volker and Pfetsch, Marc},
  journal={Mathematical Programming},
  volume={114},
  number={1},
  pages={1--36},
  year={2008},
  publisher={Springer}
}

@article{Margot2002PruningIsomorphism,
  author  = {Fran{\c{c}}ois Margot},
  title   = {Pruning by Isomorphism in Branch-and-Cut},
  journal = {Mathematical Programming},
  volume  = {94},
  number  = {1},
  pages   = {71--90},
  year    = {2002},
  doi     = {10.1007/s10107-002-0358-2}
}

@article{LibertiOstrowski2014StabilizerBased,
  author  = {Leo Liberti and James Ostrowski},
  title   = {Stabilizer-Based Symmetry Breaking Constraints for Mathematical Programs},
  journal = {Journal of Global Optimization},
  volume  = {60},
  pages   = {183--194},
  year    = {2014},
  doi     = {10.1007/s10898-013-0106-6}
}

@inproceedings{salvagnin2018symmetry,
  title={Symmetry Breaking Inequalities from the Schreier-Sims Table},
  author={Salvagnin, Domenico},
  booktitle={International Conference on the Integration of Constraint Programming, Artificial Intelligence, and Operations Research},
  pages={521--529},
  year={2018},
  organization={Springer}
}

@inproceedings{10.5555/3006433.3006560,
author = {Gent, Ian P. and Smith, Barbara M.},
title = {Symmetry Breaking in Constraint Programming},
year = {2000},
publisher = {IOS Press},
address = {NLD},
booktitle = {Proceedings of the 14th European Conference on Artificial Intelligence},
pages = {599–603},
numpages = {5},
location = {Berlin, Germany},
series = {ECAI'00}
}

@inproceedings{fahle2001symmetry,
  title={Symmetry Breaking},
  author={Fahle, Torsten and Schamberger, Stefan and Sellmann, Meinolf},
  booktitle={International Conference on Principles and Practice of Constraint Programming},
  pages={93--107},
  year={2001},
  organization={Springer}
}

@inproceedings{hubert2024preserving,
  title={Preserving and Exploiting Symmetry in Algebraic Computation},
  author={Hubert, Evelyne},
  booktitle={Proceedings of the 2024 International Symposium on Symbolic and Algebraic Computation},
  pages={13--13},
  year={2024}
}

@article{DingliwalEtAlArXiv2019SMTSymBreak,
  author    = {Saket Dingliwal and Ronak Agarwal and Happy Mittal and Parag Singla},
  title     = {Advances in Symmetry Breaking for SAT Modulo Theories},
  journal   = {arXiv preprint arXiv:1908.00860},
  year      = {2019},
  url       = {https://arxiv.org/abs/1908.00860}
}

@article{muir2024temporal,
  title={Temporal Bin Packing with Half-Capacity Jobs},
  author={Muir, Christopher and Marshall, Luke and Toriello, Alejandro},
  journal={INFORMS Journal on Optimization},
  volume={6},
  number={1},
  pages={46--62},
  year={2024},
  publisher={INFORMS}
}

@incollection{GentPetriePuget2006SymmetryCP,
  author    = {Gent, Ian P. and Petrie, Karen E. and Puget, Jean-Fran{\c{c}}ois},
  title     = {Symmetry in Constraint Programming},
  booktitle = {Handbook of Constraint Programming},
  editor    = {Rossi, Francesca and van Beek, Peter and Walsh, Toby},
  publisher = {Elsevier},
  year      = {2006},
  pages     = {329--376}
}

@inproceedings{CrawfordEtAlKR1996,
  author    = {Crawford, James M. and Ginsberg, Matthew L. and Luks, Eugene M. and Roy, Amitabha},
  title     = {Symmetry-Breaking Predicates for Search Problems},
  booktitle = {Proceedings of the Fifth International Conference on Principles of Knowledge Representation and Reasoning (KR'96)},
  publisher = {Morgan Kaufmann},
  year      = {1996},
  pages     = {148--159}
}

@inproceedings{AloulMarkovSakallahDAC2003Shatter,
  author    = {Aloul, Fadi A. and Markov, Igor L. and Sakallah, Karem A.},
  title     = {Shatter: Efficient Symmetry-Breaking for Boolean Satisfiability},
  booktitle = {Proceedings of the 40th Design Automation Conference (DAC)},
  publisher = {ACM},
  year      = {2003},
  pages     = {836--839},
  doi       = {10.1145/775832.776028}
}

@inproceedings{backofen1999excluding,
  title={Excluding Symmetries in Constraint-based Search},
  author={Backofen, Rolf and Will, Sebastian},
  booktitle={International Conference on Principles and Practice of Constraint Programming},
  pages={73--87},
  year={1999},
  organization={Springer}
}

@Book{HBofLA,
  editor =       {Hogben, Leslie},
  title =        {Handbook of Linear Algebra},
  publisher =    {CRC Press},
  year =         2014,
  series =       {Discrete Mathematics and Its Applications},
  address =      {Boca Raton},
  edition =      {second},
  doi =          {10.1201/b16113}}

@Book{GroveBenson1985,
  author =       {Grove, Larry C. AND Benson, Clark T.},
  title =        {Finite Reflection Groups},
  publisher =    {Springer Science+Business Media},
  year =         1985,
  number =       99,
  series =       {Graduate Texts in Mathematics},
  edition =      {2nd},
  doi =          {10.1007/978-1-4757-1869-0}
}

@InCollection{50Years,
  author =       {Margot, Fran{\c c}ois},
  title =        {Symmetry in Integer Linear Programming},
  booktitle =    {50 Years of Integer Programming 1958--2008},
  publisher = {Springer},
  year =         2010,
  editor =       {J{\"u}nger, Michael AND Liebling, Thomas AND Naddef,
                  Denis AND Nemhauser, George AND Pulleyblank, William          
                  AND Reinelt, Gerhard AND Rinaldi, Giovanni AND
                  Wolsey, Laurence},
  chapter =      17,
  pages =        {647--686},
  address =      {Heidelberg Dordrecht London New York},
  doi =          {10.1007/978-3-540-68279-0}}

@inproceedings{puget1993satisfiability,
  title={On the Satisfiability of Symmetrical Constrained Satisfaction Problems},
  author={Puget, Jean-Francois},
  booktitle={International Symposium on Methodologies for Intelligent Systems},
  pages={350--361},
  year={1993},
  organization={Springer}
}

@article{faenza2009extended,
  title={Extended Formulations for Packing and Partitioning Orbitopes},
  author={Faenza, Yuri and Kaibel, Volker},
  journal={Mathematics of Operations Research},
  volume={34},
  number={3},
  pages={686--697},
  year={2009},
  publisher={INFORMS}
}

@book{martello1990knapsack,
  title={Knapsack Problems: Algorithms and Computer Implementations},
  author={Martello, Silvano and Toth, Paolo},
  year={1990},
  publisher={John Wiley \& Sons, Inc.}
}

@article{meurer2017sympy,
  title={SymPy: Symbolic Computing in Python},
  author={Meurer, Aaron and Smith, Christopher P and Paprocki, Mateusz and {\v{C}}ert{\'\i}k, Ond{\v{r}}ej and Kirpichev, Sergey B and Rocklin, Matthew and Kumar, AMiT and Ivanov, Sergiu and Moore, Jason K and Singh, Sartaj and others},
  journal={PeerJ Computer Science},
  volume={3},
  pages={e103},
  year={2017},
  publisher={PeerJ Inc.}
}

@manual{ocaml,
  title        = {The OCaml System},
  author       = {Xavier Leroy and Damien Doligez and Alain Frisch and Jacques Garrigue and Didier Rémy and Jérôme Vouillon},
  year         = 2025,
  note         = {Available at \url{https://ocaml.org/manual/5.4/index.html}; Accessed on January 16th 2026},
organization = {Institut National de Recherche en Informatique et en Automatique},
  edition      = {release 5.4}
}

@misc{gurobi,
  author = {{Gurobi Optimization, LLC}},
  title = {{Gurobi Optimizer Reference Manual}},
  year = 2024,
  url = "https://www.gurobi.com"
}

@article{manual1987ibm,
  title={IBM ILOG CPLEX Optimization Studio},
  author={Manual, CPLEX User’s},
  journal={Version},
  volume={12},
  number={1987-2018},
  pages={1},
  year={1987}
}

@misc{Preiss1998ExpressionTrees,
  author       = {Preiss, Bruno R.},
  title        = {Expression Trees},
  year         = {1998},
  note         = {Accessed on January 30, 2026; \url{https://web.archive.org/web/20170119094603/http://www.brpreiss.com/books/opus5/html/page264.html}}
}

@misc{github,
  author = {Era\c{s}cu, M\u{a}d\u{a}lina AND Middeke, Johannes},
  title = {Polynomial Symmetry Breakers},
  year = 2026,
  url = "https://github.com/merascu/PolynomialSymmetryBreakers"
}

\appendix
\section{Detailed Per-Class Experimental Results}\label{apx:DetailedPerClassExperimentalResults}
 This appendix provides a fine-grained presentation of the experimental results. The figures report the same data as \autoref{fig:CombinedClasses}; however, the results are presented separately for each individual class. The plots corroborate the findings discussed in the main body of the paper.
 
\begin{figure*}[b!]
  \centering
  \input{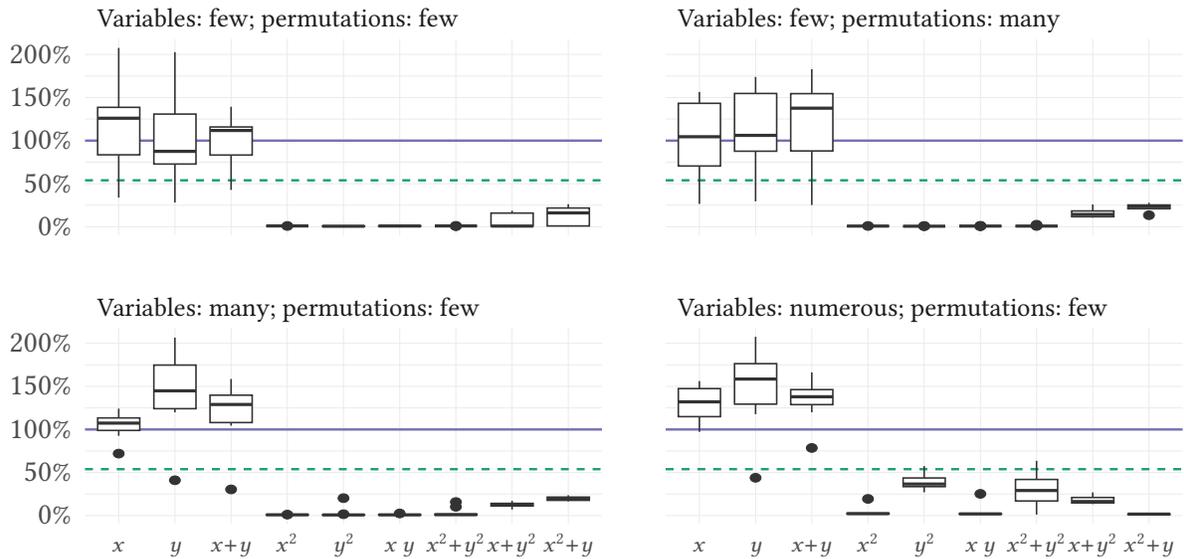}
  \caption{Results for the three classes case only.}
  \label{fig:plot3classes}
\end{figure*}

\begin{figure*}[b!]
  \centering
  \input{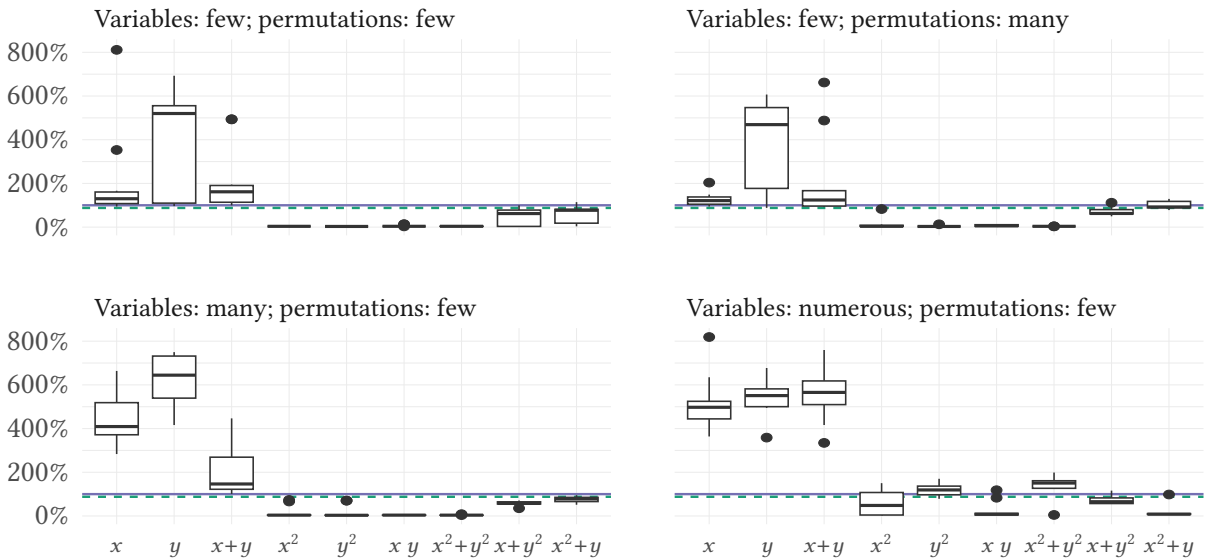}
  \caption{Results for the five classes case only.}
  \label{fig:plot5classes}
\end{figure*}

\begin{figure*}[h!]
  \centering
  \input{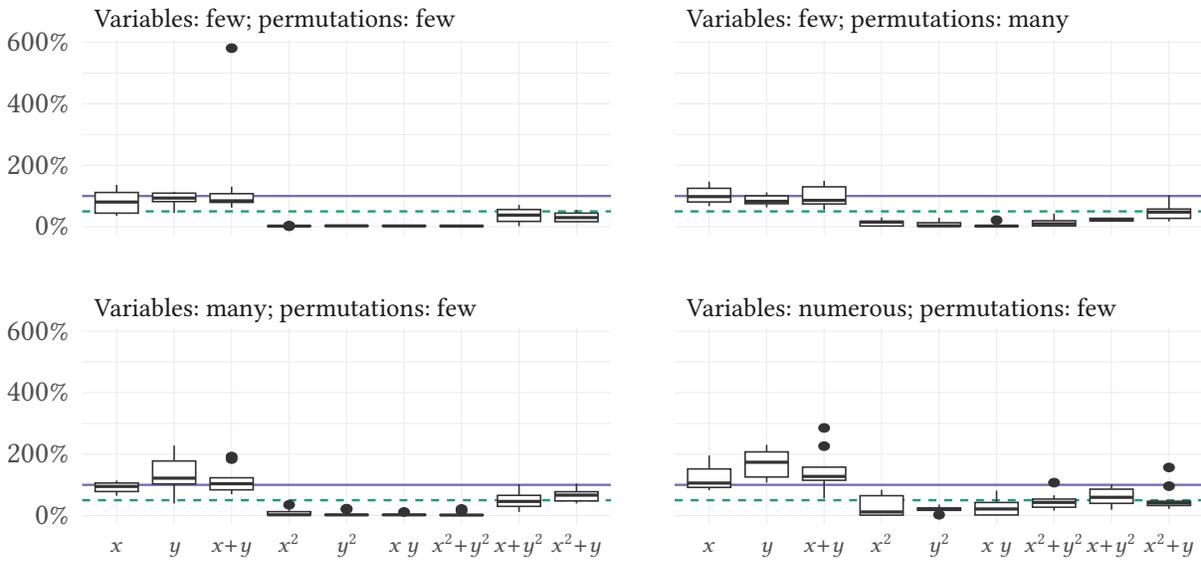}
  \caption{Results for the seven classes case only.}
  \label{fig:plot7classes}
\end{figure*}

\begin{figure*}[h!]
  \centering
  \input{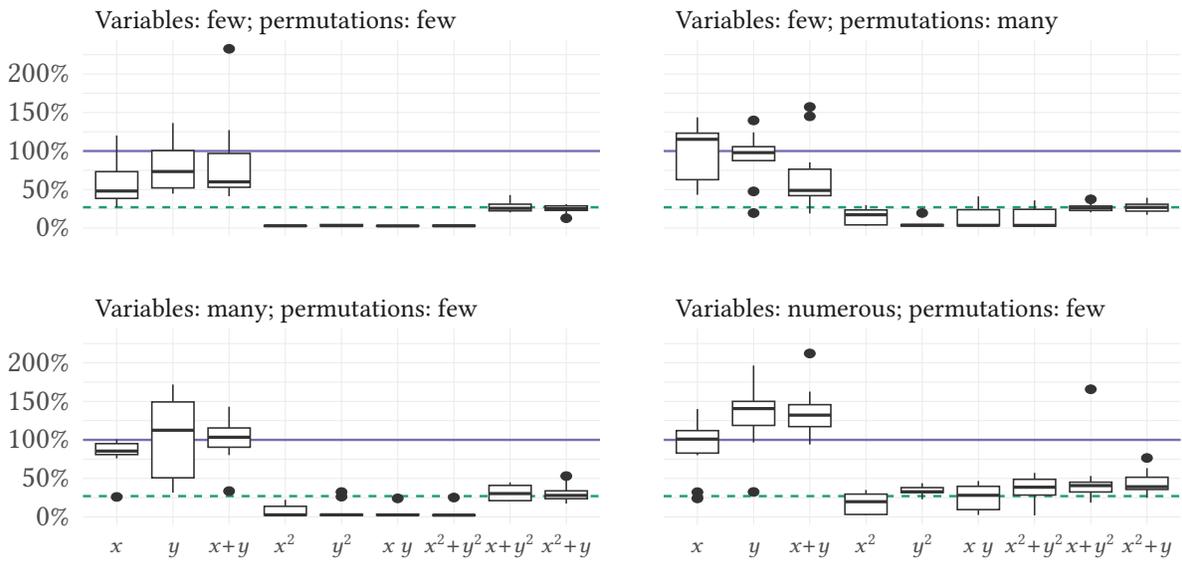}
  \caption{Results for the nine classes case only.}
  \label{fig:plot9classes}
\end{figure*}

\end{document}